\documentclass[a4paper,12pt]{article}
\usepackage{amsthm,amssymb,amsmath,url,amsfonts,mathrsfs,bm,}
\theoremstyle{definition}

\newtheorem{ex}{Example}

\newtheorem{prop}{Proposition}

\newcommand{\bal}{\begin{align*}}
\newcommand{\td}{\tilde}
\newcommand{\eal}{\end{align*}}
\renewcommand{\th}{\theta}

\newcommand{\tx}{\text}
\newcommand{\sra}{\stackrel{\leftrightarrow}}
\newcommand{\pa}{\partial}

\newcommand{\be}{\begin{eqnarray*}}

\newcommand{\na}{\nabla}
\newcommand{\mk}{\mathfrak{P}}
\newcommand{\De}{\Delta}
\newcommand{\si}{\sigma}
\newcommand{\sm}{\setminus}
\newcommand{\bDe}{\bar \Delta}
\newcommand{\lam}{\lambda}

\newcommand{\bHH}{\bar {\cal H}}
\newcommand{\HH}{{\cal H}}

\newcommand{\C}{\mathbb{C}}
\newcommand{\R}{\mathbb{R}}

\newcommand{\N}{\mathbb{N}}
\newcommand{\FF}{{\cal F}}
\newcommand{\GG}{{\cal G}}

\newcommand{\Si}{\Sigma}

\newcommand{\PP}{{\cal P}}

\newcommand{\X}{{\bf X}}

\newcommand{\1}{{\bm 1}}

\newcommand{\la}{\langle}
\newcommand{\de}{\delta}
\newcommand{\ra}{\rangle}
\newcommand{\x}{{\bf x}}

\newcommand{\y}{{\bf y}}

\usepackage{graphicx}
\begin{document}
\begin{large}
\title{\bf {Regular Hamiltonians for non-relativistic interacting quantum field theories}}
\end{large}
\author{Bruno Galvan \footnote{Electronic address: b.galvan@virgilio.it}\\ \small via Melta 16, 38121 Trento, Italy.}
\maketitle
\begin{abstract}
In the context of non-relativistic quantum field theory, a method is proposed for multiplying field operators at the same spatial point and obtaining regular (i.e. rigorously defined) interaction terms for the Hamiltonian. The basic idea is to modify the Lebesgue measure of configuration space of many particles by adding singular measures over the subspaces of configuration space in which the positions of two or more particles coincide.
\end{abstract}

\section{The problem and the proposed solution}
It is well known that quantum field theory (QFT) \cite{wein}, while provides very good empirical predictions (e.g., the value of the anomalous magnetic moment of the electron), is not formulated in a rigorous mathematical way. Actually, in the case of free particles, the theory is formulated in a rigorous way, and a regular (rigorously defined) Hamiltonian\footnote{More in general, a regular representation of the Poincar\'e group.} can be defined on the Fock space of the particles. The departure from mathematical rigor begins when interaction terms are added to the free Hamiltonian. By applying a standard quantization scheme, these terms are constructed by multiplying free fields at the same spatial point, but this construction does not give rise to regular operators. In this introductory section, this problem is illustrated in a precise way, even though in the simplified context of a non-relativistic QFT, and a possible solution is sketched for it. At the end of the section, the plan of the paper is presented.

\subsection{The problem} \label{problem}

Let $\HH:= L^2(\X, d\x)$, where $\X:= \R^3$ and $d\x$ is the Lebesgue measure on $\X$. The boson Fock space is $\FF_s(\HH)= \oplus_{i=0}^\infty \HH_s^{\otimes n}$, where $\HH_s^{\otimes n}$ is the symmetric subspace of $\HH^{\otimes n}$. For $\x \in \X$, the annihilation field operator $\phi_-(\x)$ act as follows from $\HH^{\otimes n}_s$ to $\HH^{\otimes(n-1)}_s$:
\begin{equation}
[\phi_-(\x)\Psi_n](\x_1, \ldots, \x_{n-1}):= \sqrt{n} \Psi_n(\x_1, \ldots, \x_{n-1}, \x)
\end{equation}
The operator $\phi_-(\x)$ is a regular unbounded operator defined on a suitable dense domain of $\HH_s^{\otimes n}$, for example $S C_0^\infty(\X^n)$, where $S$ is the projector on the symmetric subspace of $\HH^{\otimes n}$. Since $n$ is generic, the above equation defines the action of $\phi(\x)$ on a suitable dense domain of $\FF_0$, which is the (dense) subspace of $\FF_s(\HH)$ composed by the vectors with a finite number of non null components. The adjoint operator $\phi_-(\x)^*=:\phi_+(\x)$ acts as follows from $\HH^{\otimes n}_s$ to $\HH^{\otimes(n+1)}_s$:
\begin{equation}
[\phi_+(\x) \Psi_n](\x_1, \ldots, \x_{n+1}) = \frac{1}{\sqrt{n+1}} \sum_{i=1}^{n+1} \delta(\x - \x_i) \Psi_n(\x_1, \ldots, \hat \x_i, \ldots , \x_{n+1}).
\end{equation}
Due to the presence of the $\delta$ function, $\phi_+(\x)$ is not a regular operator, but it defines a sesquilinear form, and as a sesquilinear form it is the adjoint of $\phi_-(\x)$. However $\phi_+(\x)$ has two good properties: if $f \in \HH$, then the expression:
\begin{equation}
\phi_+(f) := \int f(\x) \phi_+(\x) d\x
\end{equation}
defines a regular operator, whose restriction to $\HH^{\otimes n}$ is a bounded operator with norm $\sqrt{n}\, ||f||$. Moreover, what is more important, also the expression
\begin{equation}
\int \phi_+(\x) \phi_-(\x) d\x 
\end{equation}
defines a regular operator, namely the number operator  $N \Psi_n = n \Psi_n$.

Let us consider now the $k$-th power of the annihilation field operator $\phi_-(\x)^k=:\phi^k_-(\x)$. It acts as follows from $\HH^{\otimes n}$ to $\HH^{\otimes (n-k)}$:
\begin{equation}
[\phi_-(\x)^k \Psi_n](\x_1, \ldots, \x_{n-k}) = \sqrt{\frac{n!}{(n-k)!}} \Psi_n(\x_1, \ldots, \x_{n-k}, \underbrace{\x, \ldots, \x}_\textrm{k times}).
\end{equation}
Also this operator is a regular operator defined on a suitable dense domain $\FF_0$. Its adjoint is
\begin{equation}
\phi^k_+(\x):=[\phi_-(\x)^k]^* =S \psi^k_+(\x),
\end{equation}
where $\psi^k_+(\x)$ acts as follows from $\HH^{\otimes n}_s$ to $\HH^{n+k}$:
\begin{equation}
[\psi^k_+(\x) \Psi_n](\x_1, \ldots, \x_{n+k}) := \sqrt{\frac{(n+k)!}{n!}} \delta(\x - \x_{n+1}) \cdots  \delta(\x-\x_{n+k})\Psi_n(\x_1, \ldots, \x_n).
\end{equation}
As before, $\phi_+^k(\x)$ is not an operator, and it is the adjoint of $\phi_-^k(\x)$ only as a sesquilinear form. However $\phi^k_+(\x)$ has no longer the good properties of $\phi_+(\x)$, namely, for generic $k, h$, the expressions 
\[
\int f(\x) \phi^k_+(\x)d\x \text{ and } \int \phi^k_+(\x)\phi^h_-(\x)d\x
\]
are no longer regular operators. Consider for example the expression
\[
\int \psi^2_+(\x)\phi_-(\x)d\x
\]
applied to a vector $\Psi_1 \in \HH$. We have that:
\begin{align*}
& \left [\left (\int \psi^2_+(\x)\phi_-(\x)d\x \right) \Psi_1 \right ](\x_1, \x_2) =  \\
& = \sqrt{2} \int \delta(\x- \x_1) \delta(\x- \x_2) \Psi_1(\x) d\x = \sqrt{2} \delta(\x_2- \x_1) \Psi_1(\x_1),
\end{align*}
which is not a regular vector. So $\int \psi^2_+(\x)\phi_-(\x)d\x$ is not a regular operator, and therefore not even $\int \phi^2_+(\x)\phi_-(\x)d\x= \int S\psi^2_+(\x)\phi_-(\x)d\x$ is a regular operator.

By quantizing classical Lagrangians one typically obtains interaction terms for the Hamiltonian containing expressions of the type $\int \phi^k_+(\x)\phi^h_-(\x)d\x$ (also in more complex contexts, namely with spin, etc...). This is basically the reason why one cannot derive a regular Hamiltonian in the context of Lagrangian QFT.

\subsection{The proposed solution} \label{solution}
The usual solution to this problem is to introduce a spatial cutoff, i.e., roughly speaking, to replace $\phi_-(\x)$ with $\phi_-(f)$, where $f \in \HH$ is strongly peaked around $\x$. In this way $\phi_+(f)$, its powers and the relative adjoint operators are regular. The problems is that a relativistic covariant theory cannot be constructed in this way.

The alternative solution proposed in this paper consists of assuming that the Hilbert space of two or more particles is not the tensor product of the Hilbert spaces of the single particles, because ``something happens'' when the particles are exactly ad the same spatial point. To be concrete, in the simple case of two particles, it is assumed that the Hilbert space of the pair is not the tensor product $\HH^{\otimes 2} = L^2(\X^2, d\x_1d\x_2)$, but rather is what will be referred to as the {\it coincidence product}:
\begin{equation} \label{10}
\HH^{\odot 2} := L^2(\X^2, [1 + \delta(\x_1 - \x_2)] d\x_1d\x_2).
\end{equation}
This means that the coincidence set $C:= \{ (\x_1, \x_2) \in \X^2: \x_1 = \x_2 \}$ has a singular measure with respect to the configuration space $\X^2$. The scalar product in $\HH^{\odot 2}$ has the following form:
\begin{equation} \label{scal}
\la \Phi |\Psi \ra = \int_{\X^2} \Phi^* \Psi [1 + \delta(\x_1 - \x_2) ]d\x_1d\x_2 =
\int_{\X^2} \Phi^* \Psi d\x_1d\x_2 + \int_\X \Phi^*(\x, \x) \Psi(\x, \x) d\x.
\end{equation}
From the above equality it is straightforward to realize that
\begin{equation} \label{dirsu}
\HH^{\odot 2} \equiv L^2(\X^2, d\x_1 d\x_2 ) \oplus L^2(\X, d\x).
\end{equation}
The two addend of the direct sum (\ref{dirsu}) will be referred to as the components of the coincidence product. 

Let us show in a simple case how the coincidence product solves the problem described in the previous section. The annihilation field operator $\phi^2_-(\x): \HH^{\odot 2} \to \C$ is defined as in the previous case:
\begin{equation}
\phi^2_-(\x)\Psi_2 = \sqrt{2} \Psi_2(\x, \x).
\end{equation}
However, in the case of a coincidence product, the adjoint $\phi^2_+(\x)$ is different, namely it is
\begin{equation} \label{ann}
[\phi^2_+(\x) c](\x_1, \x_2) = \sqrt{2} \, c \, \1_C(\x_1, \x_2) \delta(\x - \x_1),
\end{equation}
where $c \in \C$ and $\1_C$ is the characteristic function of $C$. In fact
\begin{align*}
& \la \phi^2_+(\x) c | \Psi_2 \ra = \sqrt{2} c^* \int_{\X^2}  \1_C(\x_1, \x_2) \delta(\x - \x_1) \Psi_2 (\x_1, \x_2) d\x_1 d\x_2 + \\
& + \sqrt{2} \, c^* \int_{\X^2}  \1_C(\x_1, \x_2) \delta(\x- \x_1)  \Psi_2 (\x_1, \x_2) \delta(\x_1 - \x_2) d\x_1 d\x_2 = \\
& = \sqrt{2} \, c^* \int_\X  \1_C(\x, \x_2) \Psi_2 (\x, \x_2) d\x_2 + \sqrt{2} c^* \Psi_2(\x, \x) = 
\sqrt{2} \, c^* \Psi_2(\x, \x) = \la c | \phi^2_-(\x)\Psi_2 \ra.
\end{align*}
The definition (\ref{ann}) differs from the usual definition of creation field operator for the presence of the function $\1_C$. The consequence is that now an expression of the type 
\[
\int_\x \phi^2_+(\x) \phi_-(\x) d\x
\]
defines a regular operator. In fact, for $\Psi_1 \in \HH$, we have: 
\begin{align*}
& \left [ \left (\int \phi^2_+(\x) \phi_-(\x) d\x \right ) \Psi_1 \right ] (\x_1, \x_2) = \\
& = \sqrt{2}  \int \1_C(\x_1, \x_2)  \delta(\x - \x_1) \Psi_1(\x) d\x = \sqrt{2} \Psi_1(\x_1) \1_C(\x_1, \x_2).
\end{align*}
The fact that $\Psi_2(\x_1, \x_2):= \sqrt{2} \Psi_1(\x_1) \1_C(\x_1, \x_2)$ is a regular vector of $\HH^{\odot 2}$ can be seen as follows:
\begin{align*}
& ||\Psi_2||^2 = 2 \int_{\X^2} |\Psi_1(\x_1)|^2 \1_C(\x_1, \x_2) d\x_1 d\x_2 + 2 \int_{\X^2} |\Psi_1(\x_1)|^2 \1_C(\x_1, \x_2) \delta(\x_1 - \x_2) d\x_1 d\x_2 =\\
& 2 \int_\X |\Psi_1(\x_2)|^2 \1_C(\x_2, \x_2) d\x_2 = 2 ||\Psi_1||^2.
\end{align*}

This result turns out to be general, and the fields operators $\phi^k_\pm(\x)$ turn out to possess nice properties, very similar to those of the basic fields operator $\phi_\pm(\x)$.

\subsection{The free Hamiltonian}

The replacement of the tensor product with the coincidence product requires redefining the free Hamiltonian, which in a non-relativistic context is based on the Laplacian. The problem is therefore to define a non trivial self-adjoint Laplacian on the coincidence product space, where ``not trivial'' means that it must determine a time evolution which mixes the various components the coincidence product, i.e., the addend of the direct sum (\ref{dirsu}). Also this problem is addressed in the paper, and such a kind of Laplacian is obtained by choosing a suitable domain of definition. In fact, it is well known from functional analysis that very different self-adjoint operators can be derived from the same differential operator if different domains are chosen.

\subsection{Plan of the paper}

The subjects sketched in the previous two subsections are developed in this paper for a system of non-relativistic indistinguishable bosons, and therefore the possibility that this approach could lead to a regular relativistic covariant QFT is not verified.

In Section \ref{some} a preliminary notation is introduced; in Section \ref{coincidence} the notion of coincidence product is defined in a general way; in Section \ref{field} the field operators and their powers are defined, and their properties are stated. Due to the length, the proof of these properties has been moved to the Appendix. In Section \ref{free} the problem of the definition of the Laplacian on the coincidence product space is addressed, and the Laplacian for the case of two particles is developed into the details. Section \ref{conclusion} concludes the paper by summarizing it and by presenting some open questions which require further investigation. In spite of a rather wide research in the literature, I have not found an approach which is similar to the one proposed in this paper; this is the reason for the shortness of the bibliography.

\section{Some preliminary notations} \label{some}
In order to define the coincidence product and the field operators it is useful to generalize the notion of configuration space from the usual Cartesian product $\X^n$ to the set $\X^N$, where $N$ is a generic non empty finite set of $\N$. Recall that the set $\X^N$ is composed by the functions $x_N:N \to \X$. In a less rigorous but more intuitive manner, one can say that, while an element of $\X^n$ is the $n$-uple 
\[
(\x_1, \ldots, \x_n),
\]
an element of $\X^N$, where $N= \{n_1 < \cdots < n_n\}$, is the $n$-uple
\[
(\x_{n_1}, \ldots, \x_{n_n}).
\]
In this way $\X^n$ can be identified with $\X^{\{1, \ldots, n\}}$. In this section some notation and rules relative to this formalism are introduced.

Finite subsets of $\N$ will be denoted by upper case letters $N, M, I, J, K, H$. The cardinality of $N$ will be denoted by $|N|$, the $i$-th element of $N$, in its natural order, will be denoted by $n_i$, i.e., by using the lower case of the letter denoting the set. An element of $\X^N$ will be denoted by $x_N$. According to the above intuitive representation of the elements of $\X^N$, if $i \in N$, the symbol $\x_i$ will be also used to denote $x_N(i)$.

If $I$ and $J$ are non intersecting finite subsets of $\N$, then any unordered pair $(x_I, x_J)$ univocally defines an element $x_{I \cup J} \in \X^{I \cup J}$, and vice-versa. So the sets $\X^I \times \X^J$, $\X^J \times \X^I$, and $\X^{I \cup J}$. can be identified. If $I \subseteq N$, the projector $\pi_I:\X^N \to \X^I$ is naturally defined as follows: $\pi_I x_N := x_N \big |_I$.

If $|N| = |M|$ there is a natural bijection $\theta: X^N \to X^M$, namely $[\th x_M](m_i ):= x_N(n_i)$. As a consequence, an element of $\X^M$ can be used as the argument of a function $\Psi_N:\X^N \to \C$, by defining $\Psi_N(\x_M) := \Psi_N(\th^{-1} x_N)$. In a more intuitive notation:
\[
\Psi_N(x_M) = \Psi_N(\x_{m_1}, \ldots, \x_{m_{|M|}}).
\]
In the same way, if $I \cap J = \emptyset$ and $|I \cup J| = |N|$, we can write unambiguously $\Psi_N(x_I, x_J)$ or $\Psi_N(x_J, x_I)$ to denote $\Psi_N(x_{I \cup J})$, where $x_{I \cup J}$ is the element of $\X^{I \cup J}$ associated with the pair $(x_I, x_J)$. We will also write $x_J = \x$ if $x_J (i) = \x$ for all $i \in J$, and $\Psi_N(x_I, x_J = \x)$ in place of $\Psi_N(x_I, x_J) \big |_{x_J = \x}$.

Eventually, as a general rule, if $N = \{1, \ldots, n\}$, the symbol $n$ will replace the symbol $N$ when possible. For example, we will write $\X^n, x_n$ and $\Psi_n$ in place of $\X^N, x_N$ and $\Psi_N$.

\section{The coincidence product} \label{coincidence}

In this section the coincidence product of two Hilbert spaces defined by equation (\ref{10}) will be generalized to $n$ Hilbert spaces, and the Fock space composed by these spaces will be introduced.

Let us start by defining the coincidence planes of $\X^N$. Given $I \subseteq N$, define the following subspace of $\X^N$:
\begin{equation}
C_I = \{ x_N \in \X^N: \forall i, j \in N, \; \; i , j \in I \Rightarrow \x_i = \x_j\}.
\end{equation}
Moreover, given a partition $\PP = \{I_1, \ldots, I_p\}$ of $N$, define the subspace:
\begin{equation} \label{cap}
C_\PP = C_{I_1} \cap \ldots \cap C_{I_p}.
\end{equation}
The subspaces $C_I$ and $C_\PP$ will be referred to as the {\it coincidence planes} corresponding to the subset $I$ and to the partition $\PP$, respectively\footnote{In scattering theory these subspaces are usually referred to as the collision planes.}. If we defines $F_I:= \pi_I C_I \subseteq \X^I$, we can write:
\begin{equation} \label{21}
C_\PP = F_{I_1} \times \cdots \times F_{I_p}.
\end{equation}
Note that $\PP_1 > \PP_2 \Leftrightarrow C_{\PP_1} \subset C_{\PP_2}$, and $C_{\PP_1} \cap C_{\PP_2} = C_{\PP_1 \vee \PP_2}$. 

\begin{ex}
Let us consider the two extreme partitions of $N$, namely 
\[
\PP_{\min} := \{\{n_1\}, \ldots, \{n_{|N|}\}\} \tx{ and } \PP_{\max} := \{\{n_1, \ldots , n_{|N|}\}\}.
\]
Then $C_{\PP_{\min}} = \X^N$ and $C_{\PP_{\max}}= \{ x_N \in \X^N: \x_i = \x_j \tx{ for all } i, j \in N\}$.
\end{ex}

Let us consider now the measures. The natural correspondence $\X^n \leftrightarrow \X^N$ (see Section \ref{some}) induces on $\X^N$ the Lebesgue measure, that will be denoted by $dx_N$. If $I \cap J=\emptyset$ and $I \cup J=N$, one can write $dx_N = dx_I dx_J$. Since $\X^N = \X^{\{n_1\}} \times \cdots \times \X^{\{n_{|N|}\}}$, one can also write $dx_N = dx_{\{n_1\}} \cdots dx_{\{n_{|N|}\}}$. Here too, in order to simplify the notation, the previous expression will be replaced by the simplified expression $dx_N = d\x_{n_1} \cdots d\x_{n_{|N|}}$.

There is also a natural correspondence between $F_I$ and $\X$, which allows us to endowe $F_I$ with the Lebesgue measure, that will be denoted by $d\nu_I$. The symbol $d\mu_I$ will denote the measure on $\X^I$ which equals $d\nu_I$ on $F_I$ and is null on $\X^I \sm F_I$. The measure $d\mu_I$ can be expressed as follows:
\begin{equation}
d\mu_I = \delta(\x_{i_1} - \x_{i_2}) \cdots \delta(\x_{i_{|I|-1}} - \x_{i_{|I|}}) dx_I.
\end{equation}
Due to equation (\ref{21}), $C_\PP$ is naturally endowed with the measure
\begin{equation}
d\nu_\PP:= d\nu_{I_1} \cdots d\nu_{I_p},
\end{equation}
and $d\mu_\PP$ will denote the measure on $\X^N$ which equals $d\nu_\PP$ on $C_\PP$ and is null on $\X^N \sm C_\PP$. We also have
\begin{equation}
d\mu_\PP = d\mu_{I_1} \cdots d\mu_{I_p}.
\end{equation}
Let us define therefore on $\X^N$ the {\it coincidence} measure:
\begin{equation}
d\mu^{\odot N} := \sum_{\PP \in \mk_N} d\mu_\PP,
\end{equation}
where $\mk_N$ is the set of all the partitions of $N$. Correspondently, if $\HH := L^2(\X, d\x)$, define the {\it coincidence product} 
\begin{equation}
\HH^{\odot N}:= L^2(\X^N, d\mu^{\odot N}).
\end{equation}
According to a previously mentioned convention, in the case in which $N=\{1, \ldots, n\}$, we write $n$ in place of $N$ at the exponent of $\mu, \X$, and $\HH$. So:
\begin{equation}
\HH^{\odot n}:= L^2(\X^n, d\mu^{\odot n}).
\end{equation}
The scalar product in $\HH^{\odot N}$ reads:
\begin{equation}
\la \Phi| \Psi \ra = \int_{\X^N} \Phi^* \Psi d\mu^{\odot N} = \sum_{\PP \in \mk_N} \int_{\X^N} \Phi^* \Psi d\mu_\PP.
\end{equation}
The following proposition generalize the equations (\ref{scal}) and (\ref{dirsu}):
\begin{prop} \hspace{3mm}
\begin{itemize}
\item[1.] $\la \Phi| \Psi \ra = \sum_{\PP \in \mk_N} \int_{C_\PP} \Phi^* \Psi d\nu_\PP.$
\item[2.] $L^2(\X^N, d\mu^{\odot N}) \equiv  \bigoplus_{\PP \in \mk_N} L^2(C_\PP, d\nu_\PP)$.
\end{itemize}
\end{prop}
\begin{proof} Point 1. Let us introduce the following subsets of $\X^N$:
\[
B_{\PP} := C_{\PP} \sm \bigcup_{\PP' > \PP} C_{\PP'}.
\]
One can prove that the sets $\{B_{\PP}\}_{\PP \in \mk_N}$ form a partition of $\X^N$ \cite{der}. So:
\[
\la \Phi | \Psi \ra = \sum_{\PP \in \mk_N} \sum_{\PP' \in \mk_N} \int_{B_{\PP'}} \Phi^* \Psi d\mu_{\PP}.
\]
But
\begin{equation} \label{intp1}
\int_{B_{\PP'}} \Phi^* \Psi d\mu_{\PP} = \int_{B_{\PP'} \cap C_{\PP}} \Phi^* \Psi d\mu_{\PP},
\end{equation}
because $\mu_{\PP}$ is null in the set $\X^N \sm C_{\PP}$. Now, if $\PP \neq \PP'$ there are two cases: (i) $\PP < \PP'$, and (ii) $\PP \vee \PP' > \PP'$. In the first case $C_{\PP'}$ is a proper subspace of $C_{\PP}$, and therefore it has null $\mu_{\PP}$ measure. Since $B_{\PP'} \subseteq C_{\PP'}$, the integral (\ref{intp1}) is null. In the second case
\[
B_{\PP'} \cap C_{\PP} = B_{\PP'} \cap C_{\PP} \cap C_{\PP'} = B_{\PP'} \cap C_{\PP \vee \PP'} = \emptyset,
\]
so , again, the integral (\ref{intp1}) is null. In conclusion, the integral (\ref{intp1}) is not null only if $\PP = \PP'$. But
\[
\int_{B_\PP} \Phi^* \Psi d\mu_{\PP}  = \int_{C_\PP} \Phi^* \Psi d\mu_{\PP},
\]
because the set $C_{\PP} \sm B_{\PP}=\bigcup_{\PP' > \PP} C_{\PP'}$ is the union of proper subspaces of $C_{\PP}$, and therefore it has null $\mu_\PP$ measure. Eventually, obviously:
\[
\int_{C_\PP} \Phi^* \Psi d\mu_{\PP}=\int_{C_\PP} \Phi^* \Psi d\nu_{\PP}.
\]
\begin{center}
* * *
\end{center}
Point 2. According to the above reasoning one can write
\[
\la \Phi| \Psi \ra = \sum_{\PP \in \mk_N} \int_{B_\PP} \Phi^* \Psi d\nu_\PP.
\]
Since the sets $\{B_\PP\}_{\PP \in \mk_N}$ from a partition of $\X^N$, one easily deduce that
\[
L^2(\X^N, d\mu^{\odot N}) \equiv  \bigoplus_{\PP \in \mk_N} L^2(B_\PP, d\nu_\PP|_{B_\PP}).
\]
But, since $C_{\PP} \sm B_{\PP}$ has null $d\nu_{\PP}$ measure, 
\[
L^2(B_\PP, d\nu_\PP|_{B_\PP}) \equiv L^2(C_\PP, d\nu_\PP).
\]
\end{proof}
The addends of the direct sum in point 2 of the above proposition will be referred to as the {\it components} of the coincidence product.

\begin{ex}
Let us verify if $\HH^{\odot 2}$ actually corresponds to the definition (\ref{10}). The set $\{1, 2\}$ has the following subsets: $I_1 := \{1\}$, $I_2 := \{2\}$ and $I_3 := \{1, 2\}$. It is easy to see that $C_{I_1}= C_{I_2} = \X^2$ and $C_{I_3} = \{(\x_1, \x_2) \in \X^2 : \x_1 = \x_2\}$, which is the coincidence set $C$ of section \ref{solution}. The set $\{1, 2\}$ has two partitions, namely $\PP_1= \{I_1, I_2\}$ and $\PP_2 = \{I_3\}$. So $C_{\PP_1}= C_{I_1} \cap C_{I_2} = \X^2$,  and $C_{\PP_2}=C_{I_3}$. The measures: $d\mu_{I_1}= d\x_1$ and $d\mu_{I_2}= d\x_2$, so that $d\mu_{\PP_1} = d\x_1 d\x_2$, and $d\mu_{I_3}= d\mu_{\PP_2}= \de(\x_1 - \x_2) d\x_1 d\x_2$. Eventually, 
\begin{equation}
d\mu^{\odot 2} = d\mu_{\PP_1} + d\mu_{\PP_2} = [1 + \de(\x_1 - \x_2)] d\x_1 d\x_2,
\end{equation}
which corresponds to the definition (\ref{10}).
\end{ex}

\subsection{Permutations}

Let $\Sigma_N$ denote the set of the permutations of the set $N$. If $\si \in \Sigma_N$, $I \subseteq N$ and $\PP \in \mk_N$, the expressions $\si(I)$ and $\si(\PP)$ are defined in an obvious way. A permutation $\si \in \Sigma_N$ naturally induces a bijection $f_\si: \X^N \to \X^N$, defined as follows:  $f_\si(x_N) := x_N \cdot \si^{-1}$ (recall that $x_N$ is a map from $N$ to $\X$). It is straightforward to prove that $f_{\si_1} \cdot f_{\si_2} = f_{\si_1 \cdot \si_2}$ and $f_{\si^{-1}}= f^{-1}_\si$. If $\Delta \subseteq \X^N$, the expression $f_\si(\Delta)$ is defined in an obvious way.

\begin{prop} Let $\si \in \Sigma_n$ and $\De \subseteq \X^N$; then
\begin{itemize}
\item[1.] $f_\si(C_I) = C_{\si(I)}$ and $f_\si(C_\PP) = C_{\si(\PP)}$;
\item[2.] $\mu_\PP(\De) = \mu_{\si(\PP)}[f_\si(\De)]$;
\item[3.] $\mu^{\odot N}[f_\si(\Delta)]=\mu^{\odot N}(\Delta)$.
\end{itemize}
\end{prop}
\begin{proof} Proof of 1.
\begin{align*}
f_\si(C_I) & = \{ f_\si(x_N) \in \X^N: i , j \in I \Rightarrow x_N(i) = x_N(j)\}= \\  
& = \{ x_N \in \X^N: i , j \in I \Rightarrow [f^{-1}_\si x_N](i) = [f^{-1}_\si x_N](j)\}= \\ 
& = \{ x_N \in \X^N: i , j \in I \Rightarrow x_N[\si(i)] = x_N[\si(j)]\}= \\  
& = \{ x_N \in \X^N: \si^{-1}(i), \si^{-1}(j) \in I \Rightarrow x_N(i) = x_N(j)\}= \\  
& = \{ x_N \in \X^N: i, j \in \si(I) \Rightarrow x_N(i) = x_N(j)\}= C_{\si(I)}.
\end{align*}
The transformation rule for the sets $C_\PP$ derives from the above result and from equation (\ref{cap}).

Proof of 2. Omitted.

Proof of 3. 
\[
\mu^{\odot N}[f_\si (\De)] = \sum_{\PP \in \mk_N} \mu_\PP[f_\si(\Delta)] = \sum_{\PP \in \mk_N} \mu_{\si^{-1}(\PP)} (\Delta)= \mu^{\odot N}(\De).
\]
\end{proof}
The third point of the above proposition states that the coincidence measure is invariant under permutations.

Let us define the linear operator $U_\si: \HH^{\odot N} \to \HH^{\odot N}$ as follows: $[U_\si \Psi](x_N) := \Psi [f^{-1}_\si(x_N)]$. It is easy to prove that $U_{\si_1} U_{\si_2} = U_{\si_1 \cdot \si_2}$ and $U_{\si^{-1}}= U^{-1}_\si$. Moreover, since the measure $\mu^{\odot N}$ is invariant under $f_\si$, the operator $U_\si$ is unitary, that is $U_\si^{-1}=U^*_\si$.

A vector $\Psi_N \in \HH^{\odot N}$ is said to be symmetric if $U_\si \Psi_N = \Psi$ for any $\si \in \Sigma_N$. Let $\HH^{\odot N}_s$ denote the subspace of $\HH^{\odot N}$ composed by the symmetric vectors. The orthogonal projector on $\HH^{\odot N}_s$ is defined in the usual way:
\begin{equation} \label{S}
S \Psi_N := \frac{1}{|N|!} \sum_{\sigma \in \Sigma_N} U_\si \Psi_N.
\end{equation}
It is easy to prove that the definition (\ref{S}) is correct also for the space $\HH^{\odot N}$, that is: (i) $S\Psi_N$ is symmetric and (ii) $S$ is an orthogonal projector.

Eventually let us define Fock space $\GG_s(\HH)$ based on the coincidence product spaces:
\begin{equation}
\GG_s(\HH):= \oplus_{n=0}^\infty \HH_s^{\odot n},
\end{equation}
with the usual convention that $\HH_s^{\odot 0} = \C$. The space $\GG_s(\HH)$ will be referred to as the {\it coincidence Fock space}.


\section{Field operators} \label{field}

In this section the annihilation and creation field operators will be defined on the coincidence Fock space. The good properties of these operators envisaged in section \ref{solution} will be proved in great generality. In particular, it will be proved that powers of annihilation operators and their adjoint, the creation field operators, can be composed to form regular interaction operators.

The action of the field operators will be defined on a generic space $\HH^{\odot N}_s$. This implies that their action is also defined on $\HH^{\odot n}_s$ and on the subspace of $\GG_s(\HH)$ composed by the vectors with a finite number of non-null components. This subspace is dense in $\GG_s(\HH)$, and will be denoted by $\GG_0$.

In this section $n,k,h$ are positive integer, $f, g$ vectors of $\HH$, $\Psi_N, \Phi_N$ and $\Psi_n, \Phi_n$ are vectors of $\HH^{\odot N}_s$ and $\HH^{\odot n}_s$, respectively. The conventional definition $\HH^{\odot \emptyset}:= \C$ will also be adopted. In the definition of the creation operators the important function $\1_I^N: \X^N \to \{0,1\}$ will be used; this function is the characteristic function of the set 
\begin{equation}
B_I^N = \{ x_N \in \X^N: \forall i, j \in N, \; \; i , j \in I \Leftrightarrow \x_i = \x_j\},
\end{equation}
where $I \subseteq N$. In the appendix various properties of $\1_I^N$ will be proved.

\subsection{Annihilation operators}

There are very few changes in the definition and in the properties of these operators when defined on the coincidence product rather than on the tensor product. 

For $k \leq n$ define the {\it $k$-annihilation field operator} $\phi_-^k(\x):\HH_s^{\odot n} \to \HH_s^{\odot (n-k)}$ as follows:
\begin{equation}
[\phi^k_-(\x) \Psi_n](\x_1, \ldots, \x_{n-k}) = \sqrt{\frac{n!}{(n-k)!}} \, \Psi_n(\x_1, \ldots, \x_{n-k}, \underbrace{\x, \ldots, \x}_\textrm{k times}).
\end{equation}
For $k>n$, define $\phi^k_-(\x) \Psi_n := 0$. $\phi^1_-(\x)$ is the usual annihilation field operator, and it will be denoted by $\phi_-(\x)$. Note that $\phi^k_-(\x)= [\phi_-(\x)]^k$. The operator $\phi^k_-(\x)$ is a regular operator defined on a suitable dense domain of $\HH_s^{\odot n}$, for example $S C_0^\infty(\X^n)$. 

Given $f \in \HH$, let us define the {\it k-annihilation operator} $\phi_-^k(f) $:
\begin{equation}
\phi_-^k(f) := \int f^*(\x) \phi_-^k(\x) d\x.
\end{equation}
Here too $\phi_-(f) := \phi^1_-(f)$ corresponds to the usual annihilation operator on $\HH_s^{\otimes n}$, where it is bounded, and therefore it is defined on the whole space $\HH_s^{\otimes n}$. However, for $k >1$ $\phi^k_-(f)$ is no longer bounded on $\HH_s^{\otimes n}$, though it can be defined on a dense domain. On the contrary, on $\HH_s^{\odot n}$ the operator $\phi_-^k(f)$ is bounded for any $k$, as it will be proved below. This is probably the only difference between annihilations operators defined on tensor product spaces and coincidence product spaces.

The operator $\phi_-^k(\x)$ (and therefore also $\phi_-^k(f)$) can be easily extended to operate between the generalized spaces $\HH^{\odot N} \to \HH^{\odot M}$, where $|M| = |N| - k$ (as usual, for $k > |N|$ the operator is null). We define:
\begin{equation} \label{defm}
[\phi_-^k(\x)\Psi_N](x_M) := \sqrt{\frac{|N|!}{(|N|-k)!}} \Psi_N(x_M, x_K = \x),
\end{equation}
where $K$ is any set such that $|K| = k$ and $M \cap K = \emptyset$. Note that the definition (\ref{defm}) does not depend on the specific choice of $K$.
\subsection{Creation operators}

It is more convenient to directly define the {\it k-creation field operator} $\phi^k_+(\x)$ between the generalized spaces $\HH_s^{\odot N} \to \HH_s^{\odot M}$, where now $|M| = |N| + k$. Let us define:
\begin{equation} \label{pp}
[\phi^k_+(\x) \Psi_N](x_M) = 
k! \sqrt{\frac{|N|!}{(|N|+k)!}} \sum_{K \subseteq M: |K| = k} \de(\x -\x_{i_1}) \1^M_K(x_M) \Psi_N(\pi_{M \sm K}x_M).
\end{equation}
It will be proved that $\phi_+^k(\x)= \phi_-^k(\x)^*$ as a sesquilinear form on $\GG_0$. In the case in which $k=1$, $N= \{1, \ldots, n\}$, and $M=\{1, \ldots, n+1\}$, equation (\ref{pp}) reads:
\begin{equation}
[\phi_+(\x) \Psi_n](x_{n+1})= \frac{1}{\sqrt{n+1}}\sum_{i = 1}^{n+1} \de(\x - \x_i) \1^{n+1}_{\{i\}}(x_{n+1})  \Psi_n(\x_1, \ldots, \hat \x_i, \ldots, \x_{n+1}),
\end{equation}
which differs from the usual creation operator for the presence of the factor $\1^M_{\{i\}}(x_{n+1})$. 

The {\it $k$-creation operator} $\phi^k_+(f)$ is defined by:
\begin{equation}
\phi^k_+(f) = \int f(\x) \phi^k_+(\x) d\x.
\end{equation}
It will be proved that $\phi^k_+(f)$ is a well defined operator on $\GG_0$, and that the operator 
\begin{equation}
\phi^k(f) := \phi^k_+(f) + \phi^k_-(f)
\end{equation}
is symmetric on $\GG_0$. Recall from section \ref{problem} that, when defined on tensor product spaces, $\phi^k_+(f)$ only exists as a sesquilinear form.

\subsection{Interaction operators}
The most important consequence of operating with the coincidence product spaces is that one can define the following {\it interaction} operators:
\begin{equation}
H^h_k := \int \phi_+^h(\x) \phi_-^k(\x) d\x.
\end{equation}
For $k \leq |N|$, the explicit action of $H^h_k$ from $\HH^{\odot N}_s$ to $\HH^{\odot M}_s$, where $|M| = |N| - k + h$, is the following:
\begin{equation} \label{espact}
[H^h_k \Psi_N](x_M) = h! \sqrt{\frac{|N|!}{(|N|-k+h)!}} \sum_{H \subseteq M: |H|=h} \1^M_H(x_M)  \Psi_N(\pi_{M \sm H} x_M, x_K = \x_{h_1}),
\end{equation}
where $K$ is any subset such that $|K|=k$ and $M \cap K= \emptyset$. If $k >|N|$ then $H^h_k= 0$. The above equation will be proved in the Appendix. It will be proved moreover that $H^h_k$ is a regular operator on the dense domain $\GG_0$, and that $H^k_h + H^h_k$ is symmetric (and arguably essentially self-adjoint) on $\GG_0$. The operators of the type $H^k_h + H^h_k$ are of course the natural candidates for composing the interaction Hamiltonian.

\subsection{Properties of the field operators} \label{properties}

The main properties of the above defined operators are listed in the following 

\begin{prop} \hspace{1mm}
\begin{itemize}
\item[1.] $\la \phi_+^k(\x) \Phi | \Psi \ra = \la  \Phi | \phi_-^k(\x) \Psi \ra$ for any $\Phi, \Psi \in \GG_0$.
\item[2.] The operators $\phi^k_\pm(f)$ are well defined on $\GG_0$, and the operator $\phi^k(f)$ is symmetric on $\GG_0$.
\item[3.] Commutation rules:
\begin{align*}
& [\phi^k_-(f), \phi^h_+(g)]\Psi =  \delta_{kh} h! \la f| g \ra \Psi; \\
&[\phi^k_-(f), \phi^h_-(g)]\Psi = [\phi^k_+(f), \phi^h_+(g)]\Psi = 0
\end{align*}
for any $\Psi \in \GG_0$.
\item[4.] The equation (\ref{espact}) holds true, the operators $H^k_h$ are well defined on $\GG_0$, and the operators $H_h^k + H_k^h$ are symmetric on $\GG_0$.
\end{itemize}
\end{prop}
Due to its length, the proof of this proposition has been moved to the appendix.

Points 2 and 5 derive from the fact that the operators $\phi_\pm^k(f)$ and $H^k_h$, when restricted to $\HH_s^{\odot N}$, are bounded. Note that the commutation rules 3 are similar to the commutation rules for the usual field operators. With respect to point 4, recall the $h, k \geq 1$, so that expression of the type $\int \phi^k_\pm(\x) d\x$ do not define regular operators. I make the following conjectures: (i) the operators $\phi^k(f)$ and $H_h^k + H_k^h$ are essentially self-adjoint on $\GG_0$, and (ii) the vacuum is cyclic, i.e., the set of the vectors of the form 
\begin{equation}
\phi^{k_1}(f_1) \cdots \phi^{k_n}(f_n) \Psi_0
\end{equation}
is total in $\GG_s(\HH)$, where $\Psi_0$ is the vacuum of $\GG_s(\HH)$. With respect to the conjecture (i), note that the operators $\phi^k(f)$ and $H_h^k + H_k^h$ commute with the conjugation $C \Psi = \Psi^*$, and that $C \GG_0 = \GG_0$; according to the von Neumann criterion, \cite{reed}, p. 143, they admit at least a self-adjoint extension.

\subsection{The interaction Hamiltonian}

Reasonable interaction Hamiltonians can therefore be composed by summing up terms of the type $H^k_h + H^h_k$ multiplied by appropriate coupling constants. The simplest interaction Hamiltonian is of the form:
\begin{equation}
H_I = e (H^2_1 + H^1_2),
\end{equation}
which corresponds to the emission and the absorption of a single particle.


\section{The free Hamiltonian} \label{free}

In the previous section we have seen how to build an interaction Hamiltonian $H_I$. In this section the non trivial problem of defining the free Hamiltonian $H_0$ on the coincidence Fock space will be addressed. 

The free Hamiltonian conserves the particle number, and therefore it will be of the form:
\begin{equation}
H_0 = \sum_{n=1}^\infty H_0^{(n)},
\end{equation}
where $H_0^{(n)}$ acts on the space $\HH^{\odot n}$. For a non relativist system of indistinguishable bosons of mass $m$, a reasonable choice for $H_0^{(n)}$ is
\begin{equation}
H_0^{(n)} = \frac{-\td \De^{(n)}}{2m},
\end{equation}
where $\td \De^{(n)}$ is an appropriate Laplacian operator acting on $\HH^{\odot n}$. Here ``appropriate'' means that this operator must act as the normal differential operator $\De$ at any point of every component of $\HH^{\odot n}$, but at the same time it must determine a time evolution which ``mixes'' the various components of $\HH^{\odot n}$. For this reason, in spite of its name, this kind of Hamiltonian in not really free, because determines some sort of interaction between the particles.

In this section the Laplacian $\td \De^{(n)}$ will be defined in a rather natural way only for the spaces $L^2(\X, [1 + \de(\x)] d\x)$ and $\HH^{\odot 2}$. The first space is not a coincidence product as defined in section \ref{coincidence}, but it has similar features and a simpler structure. The definition of the appropriate Laplacian on it is very useful for presenting in a simple way the basic principles of how such a kind of Laplacian can be built in the general case. The Laplacian for the generic space $\HH^{\odot n}$ will be presented in a future paper.

\subsection{The Laplacian on $L^2(\X, [1 + \de(\x)] d\x)$}

The space $\bHH:= L^2(\X, [1 + \de(\x)] d\x)$ is the Hilbert space of a single particle whose configuration space has a metric singularity at the origin. The scalar product is 
\begin{equation}
\la \Phi |\Psi \ra = \int_\X \Phi^* \Psi d\x + \Phi^*(0) \Psi(0),
\end{equation}
from which one easily deduce that 
\begin{equation}
L^2(\X, [1 + \de(\x)] d\x) \equiv L^2(\X,d\x) \oplus \C.
\end{equation}
It is well known from functional analysis that the structure of a self-adjoint operator may largely depend on the domain on which it is defined, and the definition of the required Laplacian will be based on a suitable definition of its domain. Let us consider first the component $L^2(\X, d\x)$ of $\bHH$, and define the domains:
\begin{align}
& D_s := \{ \Psi_s(\x) = u(||\x||)/||\x||: u \in C^\infty_0([0, \infty))\}; \\
& D := \{ \Psi_s + \Psi_n : \Psi_s \in D_s \tx{ and } \Psi_n \in C^\infty_0(\X \sm \{0\}) \},
\end{align}
where the subscripts $s$ and $n$ stay for ``singular'' and ``normal''. Recall that, if $A$ is an open subset of $\R^n$, $C^\infty_0(A)$ denotes the set of smooth functions from $A$ to $\C$ with compact support, and $C^\infty_0([0, \infty))$ is composed by the functions of $C^\infty_0(\R)$ restricted to $[0, \infty)$. The domain $D$ is a dense subspace of $L^2(\X, d\x)$, and $\De D \subseteq D$, where $\De$ is the usual differential operator $\sum_{i=1}^3 \pa_i \pa_i$. In fact, $\De C^\infty_0(\X  \sm \{0\}) \subseteq C^\infty_0(\X \sm \{0\})$ and, if $\Psi_s(\x) = u(||\x||)/||\x|| \in D_s$, then $[\De \Psi_s](\x) = u''(||\x||)/||\x|| \in D_s$, where $u''$ denote the second derivative of $u$.

Let us decompose the vector $\Psi(\x) = \Psi_n(\x) + u(||\x||)/||\x||$ in spherical harmonics:
\[
\Psi(r, s) = \sum_{l=0}^\infty \sum_{m=-l}^{l} \Psi_{lm}(r) Y_l^m(s),
\]
where $s \in S^2$, $Y_l^m$ are the spherical harmonics, and
\[
\Psi_{lm}(r) = \int_{S^2} Y^{m*}_l(s) \Psi(r, s) ds.
\]
Let us also introduce the functions $u_{lm}(r):= r \Psi_{lm}(r)$. It is easy to realize that, for $l > 0$, $\Psi_{lm}(r) = u_{lm}(r) = 0$ in a neighborhood of $0$, and $u(0) = u_{00}(0) Y_0^0 = u_{00}(0)/\sqrt{4\pi}$.

The operator $\De$ is not symmetric on $D$, but:
\begin{prop} For $\Phi, \Psi \in D$:
\begin{equation}
\la \Phi |\De \Psi \ra -  \la \De \Phi | \Psi \ra = w'^*_{00}(0)u_{00}(0) - w^*_{00}(0) u'_{00}(0),
\end{equation}
where $w_{lm}(r):= r \Phi_{lm}(r)$.
\end{prop}
\begin{proof}
Let $\la \Phi |\sra{\De}| \Psi \ra := \la \Phi |\De \Psi \ra -  \la \De \Phi | \Psi \ra$, and let $B_r$ a ball of radius $r$ centered at the origin. We have:
\[
\la \Phi| \sra{\De}| \Psi \ra  = \int_{\X} \Phi^* \sra{\Delta} \Psi d\x = \lim_{r \to 0} \int_{\X \setminus B_r} \Phi^* \sra{\Delta} \Psi d\x = \lim_{r \to 0} \int_{\pa B_r} \Phi^* \sra{\na} \Psi \hat n ds,
\]
where $ds$ is the surface element of $\pa B_r$, and $\hat n$ is the unit normal vector to the surface pointing towards the origin. The last equality can be derived by applying the divergence theorem; note that:
\[
\int_{\partial B_r} \Phi^* \nabla \Psi \hat n ds = \int_{\X \sm B_r} \nabla (\Phi^* \nabla \Psi) d\x = \int_{\X \sm B_r} \nabla \Phi^* \nabla \Psi d\x + \int_{\X \sm B_r} \Phi^* \Delta \Psi d\x.
\]
Then:
\begin{align*}
& \lim_{r \to 0} \int_{\pa B_r} \Phi^* \sra{\na} \Psi \hat n ds = 
 - \lim_{r \to 0 } r^2 \sum_{l,m} [\Phi^*_{lm}(r) \Psi'_{lm}(r) - \Phi'^*_{lm}(r) \Psi_{lm}(r)] = \\
& =  - \lim_{r \to 0 } r^2 [\Phi^*_{00}(r) \Psi'_{00}(r) - \Phi'^*_{00}(r) \Psi_{00}(r)] = - [w^*_{00}(r) u'_{00}(r) - w'^*_{00}(r) u_{00}(r)].
\end{align*}
\end{proof}
In order to simplify the notation, let us introduce the two operators $A, A': D \to \C$, defined as follows:
\begin{equation}
A \Psi = u_{00}(0) \tx{ and } A'\Psi = u'_{00}(0),
\end{equation}
so that 
\begin{equation}
\la \Phi |\De \Psi \ra -  \la \De \Phi | \Psi \ra = A'\Phi^* A \Psi - A\Phi^* A'\Psi.
\end{equation}
If we consider the whole space $\bHH$, the presence of the component $\C$ in $\bHH$ allows us to remove this asymmetry of $\De$, and this will determine the mixing of the components. Concretely, for $\lam = \lam^* \neq 0$, let us define the following domain of $\bHH$:
\begin{equation}
\bar D_\lam:=\{(\Psi, A \Psi/\lam) \in \bHH: \Psi \in D\}.
\end{equation}
It is easy to recognize that this domain is dense in $\bHH$. On this domain let us define the operator 
\begin{equation}
\bDe_\lam (\Psi, A \Psi/\lambda) := (\De \Psi, \lam A' \Psi).
\end{equation}
Here $\lam$ can be considered as a sort of coupling constant between the two components $L^2(\X, d\x)$ and $\C$. For simplifying the notation, the subscript $\lam$ will be omitted. One can prove that:
\begin{prop}
The operator $-\bDe$ is symmetric and positive on the domain $\bar D$.
\end{prop}
\begin{proof} Symmetry:
\begin{align*}
& \la (\Phi, A \Phi/\lam) | \sra{\bDe}|(\Psi, A \Psi/\lam) \ra = \la (\Phi, A \Phi/\lam) | (\De\Psi, \lam A' \Psi) \ra - \la (\De\Phi, \lam A' \Phi)|(\Psi, A \Psi/\lam) \ra = \\
& = \la \Phi |\sra{\De}|\Psi \ra + A \Phi^* A' \Psi - A' \Phi^* A \Psi =0.
\end{align*}
Positivity: recall that, in terms of the functions $u_{lm}$, the Laplacian on $D$ can be written as follows:
\[
[\De \Psi](r, s) = \sum_{l,m} \frac{1}{r} \left [u''_{lm} - \frac{l(l+1)}{r^2} u_{lm} \right ] Y_l^m(s)
\]
So:
\begin{align*}
& - \la (\Psi, A\Psi/\lam)|\bDe (\Psi, A\Psi/\lam) \ra = - \int \Psi^* \De \Psi d\x - A\Psi^* A'\Psi = \\
& = - \sum_{l,m} \int_0^\infty u^*_{lm} \left [u''_{lm} - \frac{l(l+1)}{r^2} u_{lm} \right ] dr  - A\Psi^* A'\Psi = \\
& = - \sum_{l,m} \left [ \int_0^\infty - \frac{l(l+1)}{r^2} |u_{lm}|^2 dr + u'_{lm} u^*_{lm} \Big |_0^\infty - \int_0^\infty |u'_{lm}|^2 dr \right ] - A\Psi^* A'\Psi  = \\
& = \sum_{l,m} \int_0^\infty \left [ \frac{l(l+1)}{r^2} |u_{lm}|^2 + |u'_{lm}|^2 \right ] dr + u^*_{00}(0)u'_{00}(0) - A\Psi^* A'\Psi = \\
& = \sum_{l,m} \int_0^\infty \left [ \frac{l(l+1)}{r^2} |u_{lm}|^2 + |u'_{lm}|^2 \right ] dr \geq 0.
\end{align*}

\end{proof}
From a theorem of functional analysis we know that a semi-bounded symmetric operator admits a privileged self-adjoint extension, the so called Friedrich extension \cite{reed}, p. 176. I make the following conjectures: (i) $\bDe$ is essentially self-adjoint on $\bar D$, and (ii) the spectrum of the self-adjoint extension is absolutely continuous and equal to $\R^+$.

Let $\td \De$ the privileged (or unique) self-adjoint extension of $\bDe$. Let us therefore define the Hamiltonian 
\begin{equation}
H := \frac{- \td \De}{2m}.
\end{equation}
In order to show that $H$ mixes the components of $\bHH$ it is sufficient to show that the vector $(0, 1) \in \bHH$ is not an eigenvector of $H$. In fact, suppose that $H$ does not mixes the components; then $e^{-i Ht} (0, 1)=(0, c(t))$, with $|c(t)|=1$ and $c(t_1)c(t_2) = c(t_1 + t_2)$; but this implies that $(0, 1)$ is an eigenvector of $H$. In order to show that $(0, 1)$ this is not an eigenvector of $H$, let us proceed as follows: let $\Psi \in D$, and suppose that $H (0, 1)= E (0, 1)$ for some $E \geq 0$. Then we have at the same time:
\begin{align*}
& \la (\Psi, A\Psi/\lam) |H (0, 1) \ra = E \la (\Psi, A\Psi/\lam) | (0, 1) \ra = E A\Psi^*/\lam;\\
& \la (\Psi, A\Psi/\lam) |H (0, 1) \ra = \la H (\Psi, A\Psi/\lam) |(0, 1) \ra = \la (-\De \Psi, -\lam A'\Psi) |(0, 1) \ra/(2m) = -\lam A' \Psi^*/(2m).
\end{align*}
This implies $2 m E A\Psi = - \lam^2 A' \Psi$ for any $\Psi \in D$, which is impossible.

\subsection{The Laplacian on $\HH^{\odot 2}$}

In this subsection the reasoning of the previous subsection will be largely repeated. In some cases also the same symbols will be utilized for the corresponding (though different) entities. The proofs of the propositions are omitted because they are very similar to the proofs of the corresponding proposition of the previous subsection.

Recall that 
\begin{equation}
\HH^{\odot 2}=L^2(\X^2, [1 + \de(\x_1 - \x_2)] d\x_1 d\x_2) \equiv L^2(\X^2, d\x_1d\x_2) \oplus L^2(\X, d\x),
\end{equation}
and the scalar product is 
\begin{equation}
\la \Phi|\Psi \ra = \int_{\X^2} \Phi^* \Psi d\x_1 d\x_2 + \int_\X \Phi^*(\x, \x) \Psi(\x,\x) d\x.
\end{equation}
Let us introduce the new coordinates 
\begin{equation}
\x := \frac{\x_1 + \x_2}{\sqrt{2}} \tx{ and } \x_D := \frac{\x_1 - \x_2}{\sqrt{2}},
\end{equation}
The meaning of these coordinates is the following: let $C$ denote the set $\{(\x_1, \x_2 ) \in \X^2: \x_1 = \x_2\}$, which is the coincidence plane of $\X^2$ relative to the partition $\{\{1,2\}\}$. One can easily verify that, given $(\x_1, \x_2) \in \X^2$, then $||\x_D||$ is the distance between $(\x_1, \x_2) $ and $C$, and $(\x, \x)$ is the point of $C$ closest to $(\x_1, \x_2)$. In the new coordinates, $C=\{(\x_D, \x): \x_D=0\}$. 

Analogously to the previous case, let us define the following two subspaces of $L^2(\X^2, d\x_1 d\x_2)$:
\begin{align}
& D^{(2)}_s := \{ \Psi_s(\x_1, \x_2) = u(||\x_D||, \x)/||\x_D||: u \in C^\infty_0([0, \infty) \times \X)\}; \\
& D^{(2)} := \{ \Psi_s + \Psi_n : \Psi_s \in D_s^{(2)} \tx{ and } \Psi_n \in C^\infty_0(\X^2 \sm C) \},
\end{align}
where $C^\infty_0([0, \infty) \times \X)$ is composed by the functions of $C^\infty_0(\R \times \X)$ restricted to $[0, \infty) \times \X$; as a consequence, if $u \in C^\infty_0([0, \infty) \times \X)$, then $u(0, \cdot) \in C^\infty_0(\X)$. Again $D^{(2)}$ is a dense subspace of $L^2(\X^2, d\x_1 d\x_2)$, and it is an invariant subspace for the differential operator $\De^{(2)} = \De_1 + \De_2$. If $\x_D$ is expressed in spherical coordinates $(r, s)$, a vector $\Psi(r, s, \x) = \Psi_n(r, s, \x) + u(r, \x)/r$ can be expanded in spherical harmonics relative to the variable $\x_D$ as follows:
\[
\Psi(r, s, \x)=\sum_{l,m} \Psi_{lm}(r, \x) Y_l^m(s)
\]
where 
\[
\Psi_{lm}(r, \x) = \int_{S^2} Y^{m*}_l(s) \Psi(r, s, \x) ds.
\]
Here too we introduce the functions $u_{lm}(r, \x):= r \Psi_{lm}(r, \x)$. In this case, for $l > 0$, we have that $\Psi_{lm}(r, \x) = u_{lm}(r, \x) = 0$ in a neighborhood of $C$, and $u(0, \x) = u_{00}(0, \x)\sqrt{4\pi}$. As a consequence, $u_{00}(0, \cdot) \in C^\infty_0(\X)$. Let us introduce the operators $A, A':D^{(2)} \to C^\infty_0(\X)$, defined as follows: 
\begin{equation}
A\Psi := u_{00}(0, \cdot) \tx{ and } A'\Psi := u'_{00}(0, \cdot),
\end{equation}
where the apex denotes the derivative relative to the first variable of the function $u_{00}$.
\begin{prop} For $\Phi, \Psi \in D^{(2)}$:
\begin{equation}
\la \Phi |\sra{\De}| \Psi \ra = \int_\X [A'\Phi^*A \Psi - A\Phi^*A' \Psi] d\x.
\end{equation}
\end{prop}
The proof is analogous to the proof of proposition 4 in the previous subsection; the ball $B_r$ is replaced by the set $C_r:=\{ (\x_D, \x) \in \X^2 : ||\x_D|| \leq r\}$.

For $\lam = \lam^* \neq 0$, let us define the domain 
\begin{equation}
\bar D^{(2)}:=\{(\Psi, A \Psi/\lam) \in \HH^{\odot 2}: \Psi \in D^{(2)}\},
\end{equation}
and on this domain let us define the operator 
\begin{equation}
\bDe^{(2)} (\Psi, A \Psi/\lambda) := (\De^{(2)} \Psi, \lam A' \Psi + \De A \Psi/\lambda),
\end{equation}
where $\De$ is the normal Laplacian on $L^2(\X, d\x)$. Note that in this case $\bDe^{(2)}$ acts directly as a normal Laplacian also on the component $L^2(\X, d\x)$ of $\HH^{\odot 2}$. Here too the domain $\bar D^{(2)}$ is dense in $\HH^{\odot 2}$, and again one can prove that: 
\begin{prop}
the operator $\bDe^{(2)}$ is symmetric and positive on the domain $\bar D^{(2)}$.
\end{prop}
I make for $\bDe^{(2)}$ the same conjectures which have been for $\bDe$. The privileged or unique self-adjoint extension of $\bDe^{(2)}$ will be denoted by $\td \De^{(2)}$, and the Hamiltonian is
\begin{equation}
H^{(2)}:= \frac{-\td \De^{(2)}}{2m}.
\end{equation}


\section{Conclusion} \label{conclusion}
In the context of non-relativistic QFT, a method has been proposed for multiplying field operators at the same spatial point and obtaining rigorously defined interaction terms for the Hamiltonian. The basic idea is to modify the space $L^2(\R^{3n}, d^{3n}x)$, which is the usual Hilbert space of a system of $n$-non-relativistic particles, by adding singular measures to the Lebesgue measure $d^{3n}x$, in correspondence of the subspaces of $\R^{3n}$ in which the positions of two or more particle coincide. The modified Hilbert space has been referred to as the {\it coincidence product}. In this new context, powers of the field operators at the same spatial point are not so singular as in the usual context, so that annihilation and creation field operators can be multiplied and integrated over 3-space giving rise to regular operators; these operators can be utilized to compose the interaction Hamiltonian.

After having built a regular interaction Hamiltonian, the non trivial issue of how to build a free Hamiltonian on the coincidence product has been addressed. In the non-relativistic domain this problem reduces to the problem of defining a non trivial Laplacian on the coincidence product, where ``non trivial'' means that the ``normal'' and the ``singular'' components of the coincidence product have to be mixed by the time evolution determined by the Laplacian. This problem can be resolved by a suitable choice of the domain of definition of the Laplacian. Only the simplest case of two particles has been developed into the details.

This paper leaves various open questions, which basically can be grouped into the following three areas. 

1. In the non-relativist domain, is necessary (i) to define the free Hamiltonian for the general case of $n$-particles and to complete various partial results relative to (ii) the self-adjointness of the free and the complete Hamiltonian and (iii) to the the structure of their spectrum.

2. It is open the question if the proposed method can be extended to the relativistic domain. Results in this sense will be possibly presented in a future specific paper.

3. If the answer to point 2 is positive, it is necessary to verify if this mathematical construction has something to do with physical reality, that is, if the empirical predictions of standard QFT can be obtained, now in a rigorous mathematical manner, in this new scheme.


\appendix
\section{Appendix}

\subsection{The function $\1^N_I$} \label{app1}

In this subsection some properties of the function $\1^N_I$ are proved. They will be used in the next subsection to prove proposition 3 in section \ref{properties}.

For $I \subseteq N$, the function $\1_I^N: \X^N \to \{0,1\}$ is the characteristic function of the set 
\begin{equation}
D^N_I:= \{ x \in \X^N: i, j \in I \Leftrightarrow x(n_i) = x(n_j)\}.
\end{equation}

\begin{prop} In this proposition, $I, J$ are non empty subsets of $N$, and $h_N$ is a $\mu^{\odot N}$-integrable function from $\X^N \to \C$.

\begin{itemize}
\item[1.] For $\si \in \Sigma_N$:
\begin{equation}
\1^N_I \cdot f^{-1}_\si =\1^N_{\si(I)}.
\end{equation}
\item[2.]
 \begin{equation}
\int_{\X^J} \1^N_I(x_{N \sm J}, x_J) h_N(x_{N \sm J}, x_J) d\mu_J = \Big \la
\begin{array}{ll}
\1^{N \sm J}_I(x_{N \sm J}) \int_{\X^J} h_N(x_{N \sm J}, x_J) d\mu_J & \text{ for } I \cap J = 0;\\
\int_{\X^J} h_N(x_{N \sm J}, x_J) d\mu_J & \text{ for } I = J; \\
0 & \text{ otherwise}.
\end{array}
\end{equation}
\item[3.]
\begin{equation}
\int 1^N_I(x_N) h_N(x_N) d\mu^{\odot N} = \int h_N(x_N) d\mu_I d\mu^{\odot N\sm I}.
\end{equation}
\item[4.] Let $I \cap J = \emptyset$; then:
\begin{equation}
\1^N_I(x_N) \1^{N \sm I}_J(\pi_{N \sm I} x_N) = \1_I^{N\sm J}(\pi_{N\sm J} x_N) \1^N_J(x_N) = \1^N_I(x_N) \1^N_J(x_N)
\end{equation}
\end{itemize}
\end{prop}
\begin{proof} 
Proof 1.
If $X$ is a generic set, $\De \subseteq X$, $\1_\De$ is the characteristic function of $\De$ and $f:X \to X$ is a bijection, then $\1_\De \cdot f^{-1} = \1_{f(\De)}$. With a reasoning analogous to that in proposition 2, point 1, one can prove that $f_\si(D^N_I) = D^N_{\si(I)}$. So:
\[
\1_I^N \cdot f^{-1}_\si = \1_{D^N_I}\cdot f^{-1}_\si = \1_{f_\si(D^N_I)} = \1_{D^N_{\si(I)}} = \1_{\si(I)}^N.
\]

\begin{center} * * * \end{center}

Proof 2. 
\begin{align*}
& \int \1^N_I(x_{N \sm J}, x_J) h_N(x_{N \sm J}, x_J) d\mu_J = \\
& \int \1^N_I(x_{N \sm J}, x_J) h_N(x_{N \sm J}, x_J) \delta(\x_{j_1} - \x_{j_2}) \cdots \delta(\x_{j_{|J|-1}} - \x_{j_{|J|}}) dx_J = \\
& \int \1^N_I(x_{N \sm J}, x_J = \x) h_N(x_{N \sm J}, x_J = \x) d\x.
\end{align*}

For any generic finite set $K \subset \N$, let us introduce the function $F_K: \X^K \times \X \to \{0, 1\}$, defined as follows:
\[
F_K(x_K, \x) = \Big \la
\begin{array}{ll}
0 & \text{ if } \x= x_K(i) \text{ for some } i \in K;\\
1 & \text{ otherwise}.
\end{array}
\]
Of course, for any $x_K \in X^K$ and any integrable function $f: \X \to \C$, we have 
\[
\int F_K(x_K, \x) f(\x)d\x = \int f(\x)d\x.
\]
If $I \cap J = \emptyset$, then $\1^N_I(x_{N \sm J}, x_J=\x) = \1^{N \sm J}_I(x_{N \sm J}) F_I(\pi_I x_{N \sm J}, \x)$. So
\begin{align*}
& \int \1^N_I(x_{N \sm J}, x_J=\x) h_N(\x_{N \sm J}, x_J=\x) d\x = \\
& \1^{N \sm J}_I(x_{N \setminus J})  \int h_N(x_{N \setminus J}, x_J=\x) d\x = \1^{N \sm J}_I(x_{N \setminus J}) \int h_N(x_{N\sm J}, x_J) d\mu_J.
\end{align*}
If $I = J$, then $\1^N_I(x_{N \sm I}, x_I=\x) = F_{N \sm I} (x_{N \sm I}, \x)$, and therefore
\[
\int \1^N_I(x_{N \sm I}, x_I=\x) h_N(\x_{N \sm I}, x_I=\x) d\x = 
 \int h_N(x_{N \setminus I}, x_I=\x) d\x = \int h_N(x_{N \sm I}, x_I) d\mu_I.
\]
The third case splits into two cases: $J \cap (N\sm I) \neq \emptyset$ and $J \subset I$. In the first case $\1^N_I(x_{N \setminus J}, \x) = 0$; in the second case $\1^N_I(x_{N \setminus J}, \x) \leq 1 - F_{I\sm J}(x_{I\sm J}, \x)$, and therefore
\begin{align*}
& \left | \int \1^N_I(x_{N \sm J}, x_J=\x) h_N(\x_{N \sm J}, x_J=\x) d\x \right | \leq \\
&  \int \1^N_I(x_{N \sm J}, x_J=\x) |h_N(\x_{N \sm J}, x_J=\x)| d\x \leq \\
&  \int [1 - F_{I\sm J}(x_{I\sm J}, \x)] |h_N(\x_{N \sm J}, x_J=\x)| d\x = 0.
\end{align*}

\begin {center} *** \end{center}
Proof 3.
\[
\int 1^N_I(x_N) h_N(x_N) d\mu^{\odot N} = \sum_{\PP \in \mk_N} \int 1^N_I(x_N) h_N(x_N) d\mu_\PP.
\]
Let $\PP = \{I_1, \ldots, I_p \}$; then
\[
\int 1^N_I(x_N) h_N(x_N) d\mu_\PP = \int 1^N_I(x_N) h_N(x_N) d\mu_{I_1} \cdots d\mu_{I_p}. 
\]
From point 2 it is easy to realize that this integral survives only if $I = I_j$ for some $j$, in which case the function $\1^N_I$ disappears. So:
\[
\int 1^N_I(x_N) h_N(x_N) d\mu^{\odot N} = \sum_{\PP' \in \mk_{N\sm I}} \int h_N(x_N) d\mu_I d\mu_{\PP'} = \int h_N(x_N) d\mu_I d\mu^{\odot N\sm I}.
\]

\begin {center} *** \end{center}

Proof 4. The table below shows, for any possible situation of a pair of indexes $i,j \in N$, the condition that the values $x_N(i), x_N(j)$ must satisfy ($=, \ne,$ none) in order the expression at the beginning of the line be equal $1$. The notation $K := N \sm (I \cup J) $ has been adopted.
\begin{small}
\begin{center}
\begin{tabular}{|l|c|c|c|c|c|c|}
\hline
&$ i, j \in I  $&$  i \in I, j \in J $&$ i \in I, j \in K $&$ i,j \in J $&$ i \in J, j \in K $&$ i, j \in K $ \\ \hline
$\1^N_I(x_N) \1^{N \sm I}_J(\pi_{N \sm I} x_N)  $&$= $&$ \neq $&$\neq $&$ = $&$ \ne $& none \\ \hline
$\1^{N\sm J}_I(\pi_{N \sm J} x_N) \1^N_J(x_N)  $&$= $&$ \neq $&$\neq $&$ = $&$ \ne $& none \\ \hline
$\1^N_I(x_N) \1^N_J(x_N) $&$= $&$ \neq $&$\neq $&$ = $&$ \ne $& none \\ \hline
\end{tabular}
\end{center}
\end{small}
As we see, the three expressions have the same conditions, so any $x_N$ gives the same value for the three expressions.
\end{proof}

\subsection{Proof of proposition 3, section \ref{properties}}
In this proof it will always be: $|N|=n, |M|=m, |K|= k, |H|=h$.

\begin{proof} Proof 1. Without loss of generality, we can assume that $N \cap K = \emptyset$, and  $M = N \cup K$. Let us prove first that $\phi_+^k(\x) = S \psi_+^k(\x)$, where $\psi_+^k(\x)$ acts from $\HH_s^{\odot N} \to \HH^{\odot M}$ as follows:
\[
[\psi^k_+(\x) \Psi_N](x_M) :=\sqrt{\frac{(n+k)!}{n!}} \de(\x - \x_{k_1}) \1^M_K(x_M) \Psi_N(\pi_N x_M).
\]
We have:
\begin{align*}
& [S\psi^k_+(\x) \Psi_N](x_M) :=\frac{1}{(n+k)!}\sum_{\si \in \Sigma_M} [\psi^k_+(\x) \Psi_N](f^{-1}_\si(x_M)) = \\
& \frac{1}{(n+k)!}\sum_{\si \in \Sigma_M}  \sqrt{\frac{(n+k)!}{n!}} \de(\x - f^{-1}_\si(\x_{k_1})) \1^M_K(f^{-1}_\si(x_M)) \Psi_N(\pi_{M \sm K} f^{-1}_\si(x_M)).
\end{align*}
Note that 
\[
\pi_{M \sm K} f^{-1}_\si (x_M)  = x_M \cdot \si \big |_{M \sm K} = x_M \big |_{\si(M \sm K)} = 
x_M \big |_{M \sm \si(K)} = \pi_{M \sm \si(K)} x_M.
\]
So, the previous calculation continues as follows:
\begin{align*}
& \ldots = \sqrt{\frac{1}{(n+k)! n!}} \sum_{\si \in \Sigma_M}   \de(\x - \x_{\si(k_1)}) \1^M_{\si(K)}(x_M) \Psi_N(\pi_{M \sm \si(K)}x_M).
\end{align*}
Note that, if $\si_1(K) = \si_2(K)$, the corresponding addends of the sum are equal. As a consequence, instead of summing over the permutations of $\Sigma_M$ one can sum over the subsets of $M$ with $k$ elements, and multiply the sum by the number of permutations of $M$ leaving a subset of $k$ elements unchanged, which is $k! n!$. In conclusion:
\[
[\phi^k_+(\x) \Psi_N](x_M) = k! \sqrt{\frac{n!}{(n+k)!}} \sum_{I \subseteq M: |I| = k} \de(\x - \x_{i_1}) \1^M_I(x_M) \Psi_N(\pi_{M \sm I} x_M).
\]

Let us prove now that $\la \phi^k_+(\x)\Phi_N| \Psi_M \ra= \la \Phi_N| \phi^k_-(\x) \Psi_M \ra$. We have:
\begin{align*}
& \la \phi^k_+(\x)\Phi_N| \Psi_M \ra = \la S \psi^k_+(\x)\Phi_N| \Psi_M \ra = \la \psi^k_+(\x)\Phi_N| S \Psi_M \ra = \la \psi^k_+(\x)\Phi_N| \Psi_M \ra = \\
& \int [\psi^k_+(\x)\Phi_N]^*(x_N, x_K) \Psi_M(x_N, x_K) d\mu^{\odot M} = \\
& \int \sqrt{\frac{(n+k)!}{n!}} \de(\x - \x_{k_1}) \1^M_K(x_N, x_K) \Phi_N^*[\pi_N(x_N, x_K)] \Psi_M(x_N, x_K) d\mu^{\odot M} = \\
& \sqrt{\frac{(n+k)!}{n!}} \int \de(\x - \x_{k_1}) \Phi_N^*(x_N) \Psi_M(x_N, x_K) d\mu_K d\mu^{\odot N} = 
\la \Phi_N |\phi^k_-(\x) \Psi_M \ra.
\end{align*}
For $\Phi= (\Phi_0, \Phi_1, \ldots )$ and $\Psi= (\Psi_0, \Psi_1, \ldots )$ belonging to $\GG_s(\HH)$, we have:
\[
\la \phi^k_+(f)\Phi| \Psi \ra = \sum_{n=0}^\infty \la \phi^k_+(f)\Phi_n| \Psi_{n+k} \ra = 
\sum_{n=0}^\infty \la \Phi_n| \phi^k_-(f)\Psi_{n+k} \ra = \la \Phi| \phi^k_-(f)\Psi \ra. 
\]

\begin{center} * * * \end{center}
Proof of 2. In order to prove that $\phi^k_\pm(f)$ are defined on $\GG_0$ It is sufficient to prove that $\phi^k_\pm(f)$, restricted to $\HH^{\odot N}$, are bounded. 

Let us prove that $\phi^k_-(f):\HH^{\odot N} \to \HH^{\odot M}$  is bounded, where $M = N \sm K$, $K \subseteq N$, and $|K |= k$. Note first of all that, for $\Psi_N \in \HH^{\odot N}$, we have:
\begin{align*}
& \infty > ||\Psi_N ||^2 = \int |\Psi_N(x_M, x_K)|^2 d\mu^{\odot N} \geq \int 1^M_K |\Psi_N(x_M, x_K)|^2 d\mu^{\odot N} = \\
& = \int |\Psi_N(x_M, x_K)|^2 d\mu_K d\mu^{\odot M} = \int |\Psi_N(x_M, x_K=\x )|^2 d\x d\mu^{\odot M}.
\end{align*}
This implies that $\int |\Psi_N(x_M, x_K=\x )|^2 d\x < \infty$. So $g(\x):= \Psi_N(x_M, x_K=\x ) \in \HH$, and, if $f \in \HH$, we have:
\begin{align*}
& \int f^*(\y) \Psi_N(x_M, y_K= \y ) f(\x) \Psi^*_N(x_M, x_K=\x ) d\x d\y = \la f |g\ra^2 \leq  \\
&  \leq ||f||^2||g||^2 = ||f||^2 \int |\Psi_N(x_M, x_K = \x)|^2 d\x.
\end{align*}
As a consequence:
\begin{align*}
& ||\phi^k_-(f)\Psi_N||^2 = \frac{n!}{(n - k)!} \int f(y_{k_1})\Psi^*_N(x_M, y_K ) f^*(\x_1) \Psi_N(x_M, x_K) d\mu_K(y_K) d\mu_K(x_K) d\mu^{\odot M} = \\
& = \frac{n!}{(n - k)!}\int f(\y) \Psi^*_N(x_M, y_K = \y ) f^*(\x) \Psi_N(x_M, x_K = \x ) d\x d\y d\mu^{\odot M} \leq \\
&  \leq \frac{n!}{(n - k)!} ||f||^2 \int |\Psi_N(x_M, x_K = \x)|^2 d\x d\mu^{\odot M} = \frac{n!}{(n - k)!} ||f||^2 \int |\Psi_N(x_M, x_K)|^2 d\mu_K d\mu^{\odot M} = \\
& = \frac{n!}{(n - k)!} ||f||^2 \int 1^N_K(x_M, x_K) |\Psi_N(x_M, x_K)|^2 d\mu^{\odot N} \leq  \frac{n!}{(n - k)!} ||f||^2 \int |\Psi_N(x_N)|^2 d\mu^{\odot N} =\\
& = \frac{n!}{(n - k)!}||f||^2 ||\Psi_N||^2. 
\end{align*}

\vspace{3mm}
Let us prove that $\phi^k_+(f):\HH^{\odot N} \to \HH^{\odot M}$ is bounded, where we assume now $M = N \cup K$ and $K \cap N = \emptyset$.
\begin{align*}
& || \phi^k_+(f)\Psi_N||^2 = || S\psi^k_+(f)\Psi_N||^2 \leq ||\psi^k_+(f)\Psi_N||^2 = \\
& = \frac{(n+k)!}{n!} \int \1^M_K(x_M) |f(\x_{k_1})|^2 |\Psi_N(\pi_N x_M)|^2 d\mu^{\odot M} = \\
= & \frac{(n+k)!}{n!} \int |f(\x_{k_1})|^2 |\Psi_N(\pi_N x_M)|^2 d\mu_K d\mu^{\odot N} = \\
& = \frac{(n+k)!}{n!} ||f||^2 ||\Psi_N||^2.
\end{align*}

Let us prove that $\phi^k(f)=\phi^k_-(f)+ \phi^k_+(f)$ is symmetric on $\GG_0$. It is sufficient to prove that $\la \phi^k_+(f) \Phi |\Psi \ra = \la \Phi |\phi^k_-(f) \Psi \ra$ for $\Phi, \Psi \in \GG_0$. We have:
\[
\la \phi^k_+(f) \Phi |\Psi \ra = \int f^*(\x) \la \phi^k_+(\x) \Phi |\Psi \ra d\x  = 
 \int f^*(\x) \la \Phi |\phi^k_-(\x) \Psi \ra d\x = \la \Phi |\phi^k_-(f) \Psi \ra.
\]

\begin{center} * * * \end{center}

Proof of 3. It is sufficient to prove the commutation relations for a generic vector $\Psi_N \in \HH^{\odot N}$, $n \geq 0$.

Proof that 
\begin{equation} \label{recr}
[\phi^k_-(f), \phi^h_+(g)] \Psi_N = \delta_{kh} h! \la f| g \ra \Psi_N.
\end{equation}
Define $m:= n + h -k$. If $m < 0$ then the commutator is null; in this case (\ref{recr}) is satisfied because $m <0$ implies $h \neq k$. For $m \geq 0$, define $M:=\{1, \ldots, m\}$ if $m >0$, and $M:=\emptyset$ if $m=0$; assume moreover $K \cap M = \emptyset$. Let us calculate the first term $\phi^h_+(g) \phi^k_-(f) \Psi_N$:
\begin{align*}
& [\phi^h_+(g) \phi^k_-(f) \Psi_N] (x_M) = \{\phi^h_+(g) [\phi^k_-(f) \Psi_N] \} (x_M) \\
& = h! \sqrt{\frac{(n-k)!}{m!}} \sum_{J \subseteq M: |J|= h} g(\x_{J_1}) \1^M_J(x_M) [\phi^k_-(f) \Psi_N] (\pi_{M \sm J} x_M) = \\
& = h! \sqrt{\frac{(n-k)!}{m!}}\sqrt{\frac{n!}{(n-k)!}} \sum_{J \subseteq M: |J|= h} g(\x_{j_1}) \1^M_J(x_M) \int f^*(\x_{k_1}) \Psi_N(\pi_{M \sm J} x_M, x_K) d\mu_K= \\ 
& = h! \sqrt{\frac{n!}{m!}} \sum_{J \subseteq M: |J|= h} g(\x_{j_1}) \1^M_J(x_M) \int f^*(\x_{k_1}) \Psi_N(\pi_{M \sm J} x_M, x_K) d\mu_K.
\end{align*}
Note that if $n < k$, then $h >m$, and the above sum has no addends, because there is no set $J \subseteq M$. This fact corresponds to the fact that, in this case, $\phi^h_+(g) \phi^k_-(f) \Psi_N = 0$.

Let us calculate now the second term $\phi^k_-(f)\phi^h_+(g) \Psi_N$.
\begin{align*}
& \{\phi^k_-(f) [\phi^h_+(g) \Psi_N]\} (x_M) = \sqrt{\frac{(n+h)!}{m!}} \int f^*(\x_{k_1}) [\phi^h_+ (g) \Psi_N] (x_M, x_K) d\mu_K = \\
& = \sqrt{\frac{(n+h)!}{m!}} \int f^*(\x_{k_1}) h! \sqrt{\frac{n!}{(n+h)!}} \sum_{J \subseteq M \cup K: |J| = h} g(\x_{j_1}) \1^{M \cup K}_J (x_M, x_K) \Psi_N [\pi_{(M \cup K) \sm J} (x_M, x_K)] d\mu_K= \\
& = h! \sqrt{\frac{n!}{m!}} \sum_{J \subseteq M \cup K: |J| = h} \int \1^{M \cup K}_J (x_M, x_K) f^*(\x_{k_1})  g(\x_{j_1}) \Psi_N [\pi_{(M \cup K) \sm J}(x_M, x_K)] d\mu_K = \ldots
\end{align*}
Now we apply the point 2 of proposition 8 (section \ref{app1}), according to which the sets $J$ which contribute to the sum are those contained in $M$ and, only in the case in which $h=k$, the set $J=K$. Note that in the latter case, also $M =\{1, \ldots, n\}$. We obtain therefore:
\begin{align*}
& \ldots = h! \sqrt{\frac{n!}{m!}} \sum_{J \subseteq M: |J| = h} \1^{M}_J(x_M) \int f^*(\x_{k_1})  g(\x_{j_1})  \Psi_N (\pi_{M \sm J}x_M, x_K) d\mu_K + \\
& + \delta_{kh} h! \int f^*(\x_{k_1})  g(\x_{k_1}) \Psi_N[\pi_M(x_M, x_K)] d\mu_K = \\
& = [\phi^h_+(g) \phi^k_-(f) \Psi_N] (x_M) + \delta_{kh} h! \Psi_N(x_M) \int f^*(\x_{k_1})  g(\x_{k_1}) d\mu_K = \\
& = [\phi^h_+(g) \phi^k_-(f) \Psi_N] (x_M) + \delta_{kh} h! \Psi_N(x_M) \int f^*(\x) g(\x) d\x = \\
& = [\phi^h_+(g) \phi^k_-(f) \Psi_N] (x_M) + \delta_{kh} h! \la f| g \ra \Psi_N(x_M),
\end{align*}
which proves the thesis.

\vspace{3mm}
Poof that:
\[
[\phi^k_-(f), \phi^h_-(g)] \Psi_N = 0
\]
If $k+h > n$ it is obvious. Else, let $M:=\{1, \ldots, n-h-k\}$, and let $H,K$ such that $M, H, K$ are mutually disjoint. Then:
\begin{align*}
& \{\phi^k_-(f)[\phi^h_-(g) \Psi_N]\}(x_M) = \sqrt{\frac{(n-h)!}{(n-h-k)!}} \int f^*(\x_{k_1})[\phi^h_-(g) \Psi_N](x_M, x_K) d\mu_K = \\
& \sqrt{\frac{(n-h)!}{(n-h-k)!}} \int f^*(\x_{k_1}) \sqrt{\frac{n!}{(n-h)!}} \int g^*(\x_{h_1}) \Psi_N(x_M, x_K, x_H) d\mu_H d\mu_K= \\
& \sqrt{\frac{n!}{(n-h-k)!}} \int f^*(\x_{k_1}) g^*(\x_{h_1}) \Psi_N(x_M, x_K, x_H) d\mu_H d\mu_K.
\end{align*}
Since this result is symmetric in $h$ and $k$, it is necessary equal to the opposite term.

\vspace{3mm}
Calculation of 
\[
[\phi^k_+(f), \phi^h_+(g)] = 0.
\]
Let $M= \{1, \ldots, n+h +k\}$.
\begin{align*}
& \{\phi^h_+(g) [\phi^k_+(f) \Psi_N]\} (x_M) = 
h! \sqrt{\frac{(n+k)!}{(n+k+h)!}} \sum_{\{I \subseteq M: |I|= h\}} g(\x_{i_1}) \1^M_I(x_M) [\phi^k_+(f) \Psi_N] (\pi_{M \sm I} x_M) = \\
& = h! \sqrt{\frac{(n+k)!}{(n+k+h)!}} \sum_{\{I \subseteq M: |I|= h\}} g(\x_{i_1}) \1^M_I(x_M) \\
& k! \sqrt{\frac{n!}{(n+k)!}} \sum_{\{J \subseteq M \sm I: |J|= k\}} f(\x_{j_1}) \1^{M \sm I}_J(\pi_{M \sm I} x_M) \Psi_N(\pi_{M \sm (I \cup J)} x_M) = \\
& = h! k! \sqrt{\frac{n!}{(n+k+h)!}} \sum_{\{I \subseteq M: |I|= h\}}  \sum_{\{J \subseteq M \sm I: |J|= k\}}g(\x_{i_1}) f(\x_{j_1})\1^M_I(x_M)  \1^M_J(x_M) \Psi_N(\pi_{M \sm (I \cup J)} x_M) = \\
& h! k! \sqrt{\frac{n!}{(n+k+h)!}} \sum_{\{I, J \subseteq M: |I|= h, |J|= k, I \cap J = \emptyset\}}g(\x_{i_1}) f(\x_{j_1})\1^M_I(x_M)  \1^M_J(x_M) \Psi_N(\pi_{M \sm (I \cup J)} x_M).
\end{align*}
Since the expression is symmetric in $i$ and $j$, it is certainly equal to the opposite expression $[\phi^k_+(f) \phi^h_+(g) \Psi_N] (x_M)$.

\begin{center} * * * \end{center}

Proof of 4. Let us evaluate the explicit form of $H^h_k:\HH_s^{\odot N} \to \HH_s^{\odot M}$, where we assume now that $M = (N \sm K) \cup H$, with $K \subseteq N$ and $H \cap (N \sm K) =\emptyset$ (as usual, for $k >n$ the operator is null). Note that $M \sm H = N \sm K$ and $m = n-k+h$. Let us evaluate first $G^h_k:\HH_s^{\odot N} \to \HH^{\odot M}$, where $G^h_k := \int \psi^h_+(\x)\phi^k_-(\x) d\x$. We have:

\begin{align*}
& [G^h_k \Psi_N](x_M) = \int \{\psi^h_+(\x)[\phi_-^k(\x) \Psi_N]\} (x_M) d\x = \\
& \sqrt{\frac{m!}{(n-k)!}} \int \delta(\x- \x_{h_1}) \1^M_H(x_M) [\phi_-^k(\x) \Psi_N](\pi_{M \sm H} x_M) d\x = \\
& \sqrt{\frac{m!}{(n-k)!}} \sqrt{\frac{n!}{(n-k)!}}\int \delta(\x- \x_{h_1}) \1^M_H(x_M) \Psi_N(\pi_{N \sm K} x_M, x_K= \x) d\x = \\
& = A \, \1^M_H(x_M)  \Psi_N(\pi_{N \sm K} x_M, x_K = \x_{h_1}),
\end{align*}
where $A := \sqrt{n! m!}/(n-k)!$. So:
\begin{align*}
& [H^h_k \Psi_N](x_M) = [S G^h_k \Psi_N](x_M) = \frac{1}{m!} \sum_{\si \in \Si_M} U_\si [G^h_k \Psi_N](x_M) = \\
& B \sum_{\si \in \Si_M} \1^M_H(f_\si^{-1}(x_M))  \Psi_N(\pi_{M \sm H} f_\si^{-1}(x_M), x_K = f_\si^{-1}(\x_{h_1})) = \\
& B \sum_{\si \in \Si_M} \1^M_{\si(H)}(x_M)  \Psi_N(\pi_{M \sm \si(H)} x_M, x_K = \x_{\si(h_1)}),
\end{align*}
where $B = \sqrt{n!}/\sqrt{m!}(n-k)!$. Again we can transform the sum over the permutation into the sum over the set $H \subset M$ and multiply for $h! (m-h)! = h! (n-k)!$. In conclusion:
\[
[H^h_k \Psi_N](x_M) = h! \sqrt{\frac{n!}{(n-k+h)!}} \sum_{H \subseteq M: |H|=h} \1^M_H(x_M)  \Psi_N(\pi_{M \sm H} x_M, x_K = \x_{h_1}).
\]
In order to prove that $H^h_k$ is defined on $\GG_0$ it is sufficient to prove that $K^h_k$, restrictred to $\HH^{\odot N}_s$, is a bounded operator. In fact:
\begin{align*}
& ||G^h_k \Psi_N||^2 = A^2 \int \1^M_H(x_M)  |\Psi_N(\pi_{N \sm K} x_M, x_K = \x_{h_1})|^2 d\mu^{\odot M} = \\
& A^2 \int |\Psi_N(\pi_{M \sm H} x_M, x_K = \x_{h_1})|^2 d\mu_H d\mu^{M\sm H} =
A^2 \int |\Psi_N(\pi_{N \sm K} x_M, x_K = \x)|^2 d\x d\mu^{\odot N\sm K} = \\
& A^2 \int |\Psi_N(\pi_{N \sm K} x_M, x_K)|^2 d\mu_K d\mu^{\odot N\sm K} = A^2 \int |\Psi_N(x_N)|^2 d\mu_K d\mu^{\odot N\sm K} = \\
& = A^2 \int \1^N_K(x_N) |\Psi_N(x_N)|^2 d\mu^{\odot N} \leq A^2 \int |\Psi_N(x_N)|^2 d\mu^{\odot N} = A^2 ||\Psi_N||^2.
\end{align*}
In order to prove that $H^h_k + H^k_h$ is symmetric on $\GG_0$, it is sufficient to prove that $\la  G^h_k\Phi_N | \Psi_M \ra = \la \Phi_N | G^k_h \Psi_M \ra$. We have:
\begin{align*}
& \la  G^h_k\Phi_N | \Psi_M \ra  = A \, \int \1^M_H(x_M)  \Phi^*_N(\pi_{N \sm K} x_M, x_K = \x_{h_1}) \Psi_M(x_M) d\mu^{\odot M}= \\
& = A \, \int \Phi^*_N(\pi_{N \sm K} x_M, x_K = \x_{h_1}) \Psi_M(x_M) d\mu_H d\mu^{\odot M \sm H}= \\
& = A \, \int \Phi^*_N(\pi_{N \sm K} x_M, x_K = \x) \Psi_M(\pi_{M \sm H} x_M, x_H = \x) d\x d\mu^{\odot M \sm H}=\\
& = A \, \int \Phi^*_N(\pi_{N \sm K} x_M, x_K ) \Psi_M(\pi_{M \sm H} x_M, x_H = \x_{k_1}) d\mu_K d\mu^{\odot N \sm K}=\\
& = A \, \int \Phi^*_N(x_N) 1^N_K(x_N) \Psi_M(\pi_{M \sm H} x_M, x_H = \x_{k_1}) d\mu^{\odot N}= \la  \Phi_N | G^k_h \Psi_M \ra.
\end{align*}
\end{proof}

\end{document}